\documentclass[a4paper,12pt]{article}
\usepackage{amsmath, amsthm, amstext}
\usepackage{amssymb, array, amsfonts}
\usepackage[cp1251,cp866]{inputenc}
\usepackage[english]{babel}
\usepackage{graphicx}
\inputencoding{cp1251}

\numberwithin{equation}{section}

\def \al{\alpha}
\def \be{\beta}

\def \de{\delta}
\def \er{\varepsilon}
\def \ze{\zeta}

\def \la{\lambda}

\def \te{\theta}
\def \ph{\varphi}
\def \oo{\omega}
\renewcommand{\l}{\left}
\renewcommand{\r}{\right}

\def \L{\Lambda}
\def \O{\Omega}

\def \C{\mathbb{C}}

\def \R{\mathbb{R}}

\def\n{\nabla}
\def\dd{\partial}
\def\div{\operatorname{div}}
\def\rot{\operatorname{rot}}
\def\1{1\!\!\!\!1}
\def\det{\operatorname{det}}
\def\dom{\operatorname{Dom}}

\def\supp{\operatorname{supp}}
\def\tr{\operatorname{tr}}

\def\re{\operatorname{Re}}
\newcommand{\<}{\langle}
\renewcommand{\>}{\rangle}

\theoremstyle{plain}
\newtheorem{theorem}{\bf Theorem}[section]
\newtheorem{lemma}[theorem]{\bf Lemma}

\theoremstyle{definition}
\newtheorem{defi}[theorem]{\bf Definition}

\theoremstyle{remark}
\newtheorem{rem}[theorem]{\bf Remark}

\renewcommand{\le}{\leqslant}
\renewcommand{\ge}{\geqslant}

\renewcommand{\qed}{\vrule height7pt width5pt depth0pt}

\title{Regularity of electromagnetic fields in convex domains}

\author{N.~Filonov, A.~Prokhorov
\thanks{The first author was supported by grant RFBR 14-01-00306, 
the second author was supported by grant NIR SPbSU 0.38.237.2014.}
}

\begin{document}
\maketitle

\begin{abstract}
In this paper the "strong"\ Maxwell operator 
defined on fields from the Sobolev space $W_2^1$, 
and the "weak"\ Maxwell operator defined on the natural domain are considered. 
It is shown that in a convex domain, and, more generally, in a domain,
which is locally $(W_3^2 \cap W^1_\infty)$-diffeomorphic to convex one,
the "strong"\ and the "weak"\ Maxwell operators coincide.
\footnote{Key words and phrases: domain of the Maxwell operator, 
convex domains, external ball condition.}
\end{abstract}
\section*{Introduction}
\subsection{Functional spaces}
Let $\O$ be a bounded domain in $\R^3$ and let $\er$ and $\mu$ be  measurable
$(3\times 3)$-matrix-valued functions in $\O$.
The functions $\er$ and $\mu$ describe the electric permittivity 
and the magnetic permeability of the medium filling the domain. 
We assume that they are real, positively definite and bounded:
\begin{equation}
\label{01}
\er(x) = \overline{\er(x)}, \quad 
\mu(x) = \overline{\mu(x)}, \quad 
0 < \er_0 \1 \le \er(x) \le \er_1 \1\,, \quad 
0 < \mu_0 \1 \le \mu(x) \le \mu_1 \1\, .
\end{equation}
The Hilbert spaces
$$
F (\O, s) = \{ u \in L_2 (\O, \C^3) :\rot u, \div (su) \in L_2\},
\quad s = \er \text{ or } \mu,
$$
endowed with the norm
$$
\|u\|_{F(\O,s)}^2 = \|\rot u\|_{L_2}^2 + \|\div(su)\|_{L_2}^2 + (su,u)_{L_2}
$$
are natural settings in studying electromagnetic waves.
We distinguish the subspaces of functions satisfying boundary conditions 
of perfect conductivity
\begin{equation*}
F (\O, \er, \tau) = \{E \in F(\O,\er) : \left.E_\tau \right|_{\dd\O} = 0 \}, 
\end{equation*}
\begin{equation*}
F (\O, \mu, \nu) = \{H \in F(\O,\mu) : \left.(\mu H)_\nu \right|_{\dd\O} = 0 \}.
\end{equation*}
Here $\tau$ and $\nu$ mean, respectively, the tangent and the normal components 
of a vector on the boundary $\dd\O$; conditions 
\begin{equation}
\label{014}
\left.E_\tau \right|_{\dd\O} = 0 \quad \text{and} \quad
\left.(\mu H)_\nu \right|_{\dd\O} = 0
\end{equation} 
are understood in the sense of integral identities.
\begin{defi}
\label{o11}
Let $w\in L_2 (\O, \C^3)$, $\rot w\in L_2 (\O, \C^3)$.
The equality 
$\left. w_\tau\right|_{\dd\O} = 0$ means that
$$
\int_\O \<w, \rot h\> dx = \int_\O \<\rot w, h\> dx \quad 
\forall \ h \in L_2(\O, \C^3) : \rot h \in L_2 (\O, \C^3) .
$$
\end{defi}

Here $\<\,.\,,\,.\,\>$ denotes the standard scalar product in $\C^3$.

\begin{defi}
\label{o12}
Let $w\in L_2 (\O, \C^3)$, $\div w\in L_2 (\O)$.
The equality 
$\left. w_\nu\right|_{\dd\O} = 0$ means that
$$ 
\int_\O \<w, \n \ph\> dx = - \int_\O \div w\,\overline{\ph}\, dx \quad 
\forall \ \ph \in W_2^1(\O) .
$$
\end{defi}

\begin{rem}
If $\dd\O$ is Lipschitz and $w \in W_2^1(\O,\C^3)$, then
these definitions of equalities 
$\left. w_\tau\right|_{\dd\O} = 0$ and 
$\left. w_\nu\right|_{\dd\O} = 0$ 
are equivalent to the definitions in the sense of traces.
\end{rem}
For
\begin{equation}
\label{04}
\er,\mu\in{W_3^1 (\O)},
\end{equation}
we introduce the subspaces of Sobolev space with appropriate boundary conditions:
\begin{eqnarray*}
W_2^1 (\O, \tau) = 
\{u \in W_2^1 (\O, \C^3) : \left.u_\tau \right|_{\dd\O} = 0 \}, \\ 
W_2^1 (\O, \mu, \nu) = \{v \in W_2^1 (\O, \C^3) : 
\left.(\mu v)_\nu \right|_{\dd\O} = 0 \}.
\end{eqnarray*}
Assumptions \eqref{01} and \eqref{04} for domains with Lipschitz boundary
ensure the implication
$$
u \in W_2^1 (\O) \quad \Rightarrow \quad s u \in W_2^1 (\O) ,
$$ 
therefore,
$$
W_2^1 (\O, \tau) \subset F (\O, \er, \tau), \quad 
W_2^1 (\O, \mu, \nu) \subset F (\O, \mu, \nu) .
$$
Our goal is to prove the coincidence of these spaces
\begin{equation}
\label{05}
F (\O, \er, \tau) = W_2^1 (\O, \tau), \quad F (\O, \mu, \nu) = W_2^1 (\O, \mu, \nu) 
\end{equation}
in domains locally $(W^2_3\cap W^1_\infty)$-diffeomorphic to convex domains 
(see Definition \ref{o03} below).
\subsection{Maxwell operator} 
Let us distinguish the subspace
$$
J=\{ E \in L_2 (\O, \C^3) : \div (\er E) = 0 \} \oplus 
\{ H \in L_2 (\O, \C^3) : \div (\mu H) = 0 , \left.(\mu H)_\nu\right|_{\dd\O} = 0\}
$$
in the space
$L_2 (\O, \C^6; \er, \mu)$ endowed with the norm
$$
\l\|\left( \begin{array}{cc} E \\ 
H \end{array} \right)\r\|_{L_2 (\O, \C^6; \er, \mu)}^2 = 
\int_\O \l( \<\er E, E\> + \<\mu H, H\> \r) \, dx .
$$
The Maxwell operator acts on the subspace $J$ by the formula
\begin{equation}
\label{06}
{\cal M} 
\left( \begin{array}{cc} E \\ H \end{array} \right) =
\left( \begin{array}{cc}
i \er^{-1} \rot H \\ -i \mu^{-1} \rot E
\end{array} \right) 
\end{equation}
on the domain 
$$
\dom {\cal M} = 
\{E \in F(\O, \er, \tau) : \div (\er E) = 0\} \oplus 
\{H \in F(\O, \mu, \nu) : \div (\mu H) = 0 \}.
$$
Here $E$ and $H$ are electric and magnetic components of field,
subject to the divergence free condition
\begin{equation}
\label{07} 
\div (\er E) = \div (\mu H) = 0
\end{equation}
and to the boundary conditions of perfect conductivity \eqref{014}.
It is easy to show (see \cite{BS89aa}), 
that the Maxwell operator is self-adjoint, ${\cal M} = {\cal M}^*$.
One introduces also the "strong"\ Maxwell operator ${\cal M}_s$,
determined by the same expression \eqref{06} on the domain 
$$
\dom {\cal M}_s = 
\{E \in W_2^1 (\O, \tau) : \div (\er E) = 0\} \oplus 
\{H \in W_2^1 (\O, \mu, \nu) : \div (\mu H) = 0 \}.
$$
It is well known that the "strong"\ Maxwell operator does not 
always coincide with the "weak" one:
if the domain $\O$ is a polyhedron with an incoming edge, $\er = \mu = \1$\,,
then the "strong"\ operator ${\cal M}_s$ is symmetric, 
but not self-adjoint, and has infinite deficiency indices.

One consider the extended Maxwell operator ${\cal L}$ 
besides the operators ${\cal M}$ and ${\cal M}_s$.
It acts on the space 
$L_2 (\O, \C^8; \er, \mu)$ with the norm
$$
\l\|X\r\|_{L_2 (\O, \C^8; \er, \mu)}^2 = 
\int_\O \l( \<\er E, E\> + |\ph|^2 + \<\mu H, H\> + |\eta|^2\r) \, dx, 
\quad X = \left( \begin{array}{cc} E \\ \ph \\ H \\ \eta\end{array} \right).
$$
The operator ${\cal L}$ is determined by the formulas
\begin{equation}
\label{08}
{\cal L} = 
\left( \begin{array}{cccc}
0 & 0 & i \er^{-1} \rot & i\n \\
0 & 0 & - i \div (\mu \cdot) & 0 \\
-i \mu^{-1} \rot & -i\n & 0 & 0 \\
i \div (\er \cdot) & 0 & 0 & 0
\end{array} \right),
\end{equation}
$$
{\cal L} 
\left( \begin{array}{cc} E \\ \ph \\ H \\ \eta \end{array} \right) =
\left( \begin{array}{cc}
i \er^{-1} \rot H + i \n\eta \\ -i \div (\mu H) \\
-i \mu^{-1} \rot E - i \n \ph \\ i \div (\er E)
\end{array} \right) 
$$
on the domain 
$$
\dom {\cal L} = F(\O, \er, \tau) \oplus W_2^1 (\O, \C) 
\oplus F (\O, \mu, \nu) \oplus \mathring W_2^1 (\O, \C) .
$$
It is easy to see that the operator ${\cal L}$ is also self-adjoint, 
${\cal L} = {\cal L}^*$.

The subspace 
$$
\{ E \in L_2 (\O, \C^3) : \div (\er E) = 0 \} \oplus \{ 0 \} \oplus 
\{ H \in L_2 (\O, \C^3) : \div (\mu H) = 0 , \left.(\mu H)_\nu \right|_{\dd\O} = 0\}
\oplus \{ 0 \} 
$$
reduces the operator ${\cal L}$, 
and the restriction of the operator ${\cal L}$ on this subspace is unitarily equi\-va\-lent to the operator ${\cal M}$. 
The functions from $\dom {\cal L}$ admit a multiplication 
by smooth cut-off functions, 
but divergence free conditions \eqref{07} break under such a multiplication,
so it is more convenient to work with the operator ${\cal L}$.

The "strong"\ operator ${\cal L}_s$ is determined by the expression \eqref{08} 
on the domain 
$$
\dom {\cal L}_s = W_2^1 (\O, \tau) \oplus W_2^1 (\O, \C) 
\oplus W_2^1 (\O, \mu, \nu) \oplus \mathring W_2^1 (\O, \C) .
$$
The operator ${\cal L}_s$, as opposed to ${\cal M}_s$, is elliptic.
We will show (see Theorem \ref{t11}  below) that under assumptions 
\eqref{01} and  \eqref{04} in domains 
locally  $(W^2_3\cap W^1_\infty)$-diffeomorphic to convex domains,  
the equality ${\cal L} = {\cal L}_s$ (and hence, ${\cal M} = {\cal M}_s$) 
holds, as well as the equalities \eqref{05} .

\subsection{Domains}
In this subsection we describe the classes of the domains under consideration.

\begin{defi}
A domain $\L \subset \R^n$ is called special Lipschitz domain, if 
$$
\L = \{ (x', x_n) \in \R^n : x_n > \phi (x') \} ,
$$
where $\phi : \R^{n-1} \to \R$ is a Lipschitz function,
$$
|\phi (x') - \phi (y')| \le K |x'-y'| \quad \forall\ x', y' \in \R^{n-1} ,
$$
 and $K$ is the Lipschitz constant of the domain $\L$.
\end{defi}

\begin{defi}
\label{o02}
If for every point $x \in\dd\O$ there exists a neighborhood  
$D$ of the point $x$ and a special Lipschitz domain $\Lambda$, 
such that $D\cap\O = D\cap \Lambda$, 
then a bounded domain $\O \subset \R^n$ 
is called domain with Lipschitz boundary. 
\end{defi}

\begin{defi}
\label{o03}
Let $X(D,\R^n)$ be a space of functions defined in a domain $D\subset \R^n$. 
We say that the bounded domain $\O \subset \R^n$ 
belongs to the class $\mathcal{C}(X)$, if for every point 
$x \in\dd\O$ there exists a neighborhood $U$ of the point $x$, a bijection 
$$
\psi:U\to \tilde{U},\quad \psi\in X(U),\quad \psi^{-1}\in X(\tilde{U}),
$$
and a special Lipschitz domain $V$, 
such that $\psi(U\cap \O)= \tilde{U}\cap V$,
and the set $\tilde{U}\cap V$ is {\it convex}. 
\end{defi}

It is clear that the convex domains and 
the domains with boundary of class $X$ belong to $\mathcal{C}(X)$. 

Further, consider $\O \subset \R^n$, $x \in \dd\O$. 
Denote the open ball of radius $R$ centered at the point $x$ by $B_R(x)$.
Denote
$$
R(x) = \sup \{ R : \exists\, z \in \R^n
\text{ such that } |x-z| = R,\ \  B_R (z)\cap \O = \emptyset \} 
$$
to be the radius of the biggest ball, 
which can touch the point $x$ from outside the domain $\O$,
if such balls exist, and set $R(x) = 0$ if such balls do not exist.

\begin{defi}
\label{o04}
A domain $\O \subset \R^n$ is said to satisfy external ball condition 
(EBC) if 
$$
R(\O) : = \inf _{x\in\dd\O} R(x) > 0 .
$$
\end{defi}

Examples:
\begin{itemize}
\item Every convex domain $\O$ satisfies EBC; 
moreover, $R(\O)= +\infty$.
\item Every bounded domain with $C^2$-smooth boundary satisfies EBC.
\item Corner on the plane, described in the polar coordinates by the formula\\
$\O = \{ (\rho, \te) : \te \in (\pi/2, 2\pi) \}$ 
does not satisfy EBC, because $R(0) = 0$.
\end{itemize}
It turns out that external ball condition 
can be described in terms of diffeomorphisms.

\begin{theorem}
\label{t01}
For bounded domains with Lipschitz boundary, EBC is equivalent 
to belonging to the class $\mathcal{C}(C^2)$.
\end{theorem}

We will prove this theorem in \S \ref{p21}.

\subsection{Acknowledgements}
The authors are grateful to prof. A.~I.~Nazarov 
for informative comments to the paper.

\section{Regularity of electromagnetic fields}
\subsection{Statement of the result}

\begin{theorem}
\label{t11}
Let $\O$ be a bounded domain in $\R^3$, 
$\O\in\mathcal{C}(W^2_3\cap W^1_\infty)$.
Let $\er$, $\mu$ be matrices-functions satisfying 
conditions \eqref{01} and \eqref{04}.
Then the "weak"\ Maxwell operators coincide with the "strong" Maxwell operators,
${\cal L} = {\cal L}_s$, ${\cal M} = {\cal M}_s$,
the equalities \eqref{05} and the estimates
\begin{equation}
\label{11} 
\|E\|_{W_2^1} \le C \|E\|_{F(\O,\er)} \quad \forall \ E \in F(\O,\er,\tau) ,
\quad  
\|H\|_{W_2^1} \le C \|H\|_{F(\O,\mu)} \quad \forall \ H \in F(\O,\mu,\nu)
\end{equation}
hold. 
Thus, the $F$-norm and the $W_2^1$-norm in these spaces are equivalent.
\end{theorem}
\begin{rem}
The bounded domains with Lipschitz boundary,
satisfying EBC meet the requirements of the Theorem (see Theorem \ref{t01}).
Our class $\mathcal{C}(W^2_3\cap W^1_\infty)$ is bigger.
For example, the domain 
$$
\O= \{ (x', x_3) \in \R^3 : -|x'|^{3/2} < x_3 < 1 -|x'|^{3/2},|x'|<1\}
$$
does not satisfy EBC, because $R(0)=0$, but belongs to 
$\mathcal{C}(W^2_3\cap W^1_\infty)$; 
one can take the mapping 
$(x',x_3)\to (x',x_3+|x'|^{3/2})$ as a bijection $\psi$.
\end{rem}
\begin{rem}
The assumptions on the coefficients $\er$, $\mu$ and 
the boundary $\dd\O$ can be slightly relaxed.
As follows from the proof, it is enough to require \eqref{01} 
and the next condition in the spirit 
of the theory of multipliers (see \cite{MSh}):
{\it for every positive $\de$ there exists a number $C(\de)$, such that}
\begin{equation}
\label{125}
\int_\O |\dd_j s|^2 |u|^2 dx \le 
\de \|\n u\|_{L_2 (\O)}^2 + C(\de) \|u\|_{L_2 (\O)}^2, 
\quad u \in W_2^1 (\O), \quad s = \er \ \text{or}\ \mu .
\end{equation}
For $s \in W_3^1 (\O)$ the condition \eqref{125} is satisfied 
(see below Lemma \ref{l32}).
For diffeomorphisms $\psi$ which map the domain $\O$ locally 
to the convex domain, 
it is sufficient to require $\psi \in W_\infty^1$ 
and the property \eqref{125} for $s = \n\psi$.
\end{rem}

For simplicity, we formulate and prove Theorem \ref{t11} 
only for bounded domains. 
It is clear from the proof, that it can be extended 
to unbounded domains with appropriate modification 
of the assumptions on the coefficients.
We now state the case of the operator 
with periodic coefficients in the infinite cylinder.

\begin{theorem}
\label{t11'}
Let $\O = U \times \R$ be a cylinder, $U \subset \R^2$ be 
a bounded domain, $U \in \mathcal{C}(W^2_p)$, $p>2$.
Let $\er$, $\mu$ satisfying \eqref{01} 
be periodic along the axis of the cylinder with period $a$ 
and $\er, \mu \in W_3^1 (U \times [0,a])$.
Then ${\cal L} = {\cal L}_s$, ${\cal M} = {\cal M}_s$ 
and the relations \eqref{05} and \eqref{11} hold.
\end{theorem}

\begin{rem}
Assumptions on the boundary of the domain were relaxed 
to $\mathcal{C}(W^2_p)$, $p>2$,
since corresponding diffeomorphisms depend only on two variables;
it is easy to see that if 
$s = s(x_1, x_2)$, $s \in W_p^1 (U)$, $p>2$, 
then the requirement \eqref{125} is fulfilled for $\O = U \times [0,a]$.
\end{rem}

\subsection{Comments}
It is natural to consider the Maxwell operator on manifolds 
of arbitrary dimension, see \cite{Weck, DF}.
In the present paper we consider only the case of the domain in $\R^3$.
Inequalities \eqref{11} for $\er = \mu = \1$, written in the language of 
differential forms, are known as the Gaffney-Friedrichs inequalities. 
They were proved (without connection with the Maxwell operator) 
respectively by Gaffney for manifolds without boundary \cite{Ga} 
and by Friedrichs in the case of smooth manifolds with boundary \cite{Fr}.
It should be noted that the inequalities \eqref{11}, 
established only for fields in $W_2^1$, {\it do not imply} 
that the spaces \eqref{05} coincide: 
for example, in a polyhedron with an incoming edge, \eqref{11} 
holds for $\er = \mu = \1$, but \eqref{05} fails.

Without claiming completeness of the review, 
we list some known results for bounded domains. 
R.~Leis proved in 1968 equalities \eqref{05} for 
$\dd\O \in C^3$, $\er, \mu \in C^5 (\overline\O)$ \cite{L68}.
J.~Gobert obtained in 1971 the result 
(in multidimensional case) for $\dd\O \in C^2$, $\er = \mu = \1$ \cite{Go}, 
and C.~Weber in 1981 for 
$\dd\O \in C^2$, $\er, \mu \in C^1 (\overline\O)$ 
(see also the next subsection) \cite{Weber}.
In 1982 J.~Saranen proved "electric"\ equality \eqref{05}
in the case of convex domain $\O \subset \R^3$ for 
$\er \in \operatorname{Lip} (\overline\O)$ \cite{Sa}.
"Magnetic"\ equality \eqref{05} for convex domains was established 
by F.~Kikuchi and S.~Kaizu in 1986 for $\mu = \1$ \cite{KaKi}.
In 2001 M.~Mitrea found out \eqref{05} for Lipschitz domains, 
satisfying EBC, for $\er = \mu = \1$ \cite{Mi}.
In a recent paper \cite{AC} G.~Alberti and Y.~Capdeboscq investigated 
regularity of solutions of Maxwell system with nonselfadjoint $\er$ and $\mu$.
In particular, they established equalities \eqref{05} 
in bounded domains with $C^{1,1}$ smooth boundary 
under the following assumptions about the coefficients:
$$
\er + \er^* \ge \er_0 \1 > 0, \quad 
\mu + \mu^* \ge \mu_0 \1 > 0, \qquad \er, \mu \in W^1_p (\O), \ p >3 .
$$

Thus, one can say that the "electric"\ case 
of Theorem \ref{t11} was practically known. 
On the contrary, "magnetic"\ case for varying coefficients $\mu$ 
in the convex domain was not analyzed.
For convenience, we give one proof for both cases.

It should be noted that the question of compactness of embeddings
$$
F(\O,\er,\tau) \subset L_2(\O) \quad \text{and} \quad 
F(\O,\mu,\nu) \subset L_2(\O) 
$$
was also actively studied.
This compactness provides Fredholm solvability for corresponding problems 
and discreteness of spectra of the corresponding operators.
If the equalities \eqref{05} hold, then the embeddings are certainly compact.
But compactness of this embeddings takes place already 
for arbitrary Lipschitz manifolds with boundary 
(see \cite{P} and the references therein).

\subsection{Result in the smooth domain}
In this paragraph we follow the papers \cite{BS87smzh, BS89aa}.
We introduce the spaces of gradients of solutions of scalar elliptic problems
$$
E (\O,\er,\tau) = \{ \n\ph : \ph \in \mathring W_2^1 (\O),\ 
\div (\er\n\ph) \in L_2 (\O) \},
$$
$$
E (\O,\mu,\nu) = \{ \n\eta : \eta \in W_2^1 (\O),\
\div (\mu\n\eta) \in L_2 (\O), \left.(\mu\n\eta)_\nu\right|_{\dd\O} = 0 \}.
$$
These spaces are subspaces of $F(\O,\er,\tau)$ and $F(\O,\mu,\nu)$
respectively.
For bounded domains with Lipschitz boundary, 
the following decomposition takes place
\begin{equation}
\label{13}
F(\O,\er,\tau) = W_2^1(\O,\tau) + E(\O,\er,\tau).
\end{equation}
In \cite{BS87smzh} it was proved for the case $\er = \1$, 
but the proof works without changes for $\er \in W_3^1 (\O)$.

The similar decomposition for magnetic fields
\begin{equation}
\label{14}
F(\O,\mu,\nu) = W_2^1(\O,\mu,\nu) + E(\O,\mu,\nu)
\end{equation}
is not always true.
In \cite{BS87zns} it was shown that this decomposition holds
for $\mu \in C^1 (\overline\O)$ in "domains with edges and vertices";
see \cite{BS87zns} for the precise description of the class of domains; 
it includes all domains locally $C^2$-diffeomorphic to polyhedrons.
In \cite{F}, \eqref{14} was shown in the case of
$\dd\O \in C^{3/2+\epsilon}$, $\epsilon > 0$, $\mu \in W_3^1 (\O)$.
Also the domain $\O$ with the boundary of class $C^{3/2}$ was constructed there,
for which the equality \eqref{14} does not hold (for $\mu = \1$). 
It should be noted also, that equalities \eqref{13}, \eqref{14} take place
for "domains with screens"\ \cite{BS93, F96}.

If the equality \eqref{13} (resp. \eqref{14}) holds,
then possible singularities of functions from the class $F(\O,\er,\tau)$
(resp. $F(\O,\mu,\nu)$) are reduced to the singularities 
of functions from $E(\O,\er,\tau)$ (resp. $E(\O,\mu,\nu)$), 
i.e. gradients of the solutions of the Dirichlet (resp. Neumann) problem
for scalar elliptic equation of the second order 
with right hand side from $L_2(\O)$.
Scalar elliptic problems are well studied.
Strong solvability is known both for smooth and for convex domains, 
or more generally, for domains of class 
${\mathcal C} (W^2_3\cap W^1_\infty)$ in our terminology
(see \cite{LU} and \cite{Gr})
\footnote{As a matter of fact, it was the analogy 
with scalar equations of second order
that stimulated our confidence that in convex domains 
the "weak" Maxwell operator coincides with the "strong" one.}.
The inclusions 
$$
E(\O,\er,\tau) \subset W_2^1 (\O,\tau)\quad \text{and} \quad
E(\O,\mu,\nu) \subset W_2^1 (\O,\mu,\nu)
$$
follow from the strong solvability of scalar problems for 
$\er, \mu \in W_3^1 (\O)$.
These results, in particular, imply

\begin{theorem}
\label{localsmooth}
Let $\O \subset \R^3$ be a bounded domain, $\dd\O \in C^2$, 
$\er, \mu$ 
be matrices-functions satisfying conditions \eqref{01} and \eqref{04}. 
Then the equalities \eqref{05} and the estimates \eqref{11} hold.
\end{theorem}

\subsection{Plan of the paper}
Using localization and admissible diffeomorphisms, 
we reduce the proof of Theorem \ref{t11} 
to the case of special Lipschitz domain.
Considering this case, we follow the idea of \cite{Mi}:
model domains are approximated by smooth ones (\S 2),
in smooth domains uniform a priori estimates are proved (\S 4),
where the crucial point was the algebraic Lemma \ref{l31} (\S 3).
It should be noted, that the estimate \eqref{11}, 
established in Theorem \ref{localsmooth} 
for smooth domains, is not sufficient for our purposes.
We need the estimate, which is {\it uniform} with respect to 
the smooth domain, approximating the convex nonsmooth domain.
Finally, in \S 5 Theorem \ref{t11} is deduced from a priori estimates. 

\section{Classes of domains}
\subsection{Domains, satisfying EBC}
\label{p21}
 
\begin{lemma}
\label{l21}
Let $\O$ be a special Lipschitz domain, 
\begin{equation}
\label{20}
\O = \{ (x', x_n) \in \R^n : x_n > \phi (x') \} ,
\end{equation}
\begin{equation}
\label{21}
|\phi (x_1') - \phi (x_2')| \le K |x_1'-x_2'| 
\quad \forall\ x_1', x_2' \in \R^{n-1} .
\end{equation}
If $\O$ satisfies EBC on the part of the boundary 
$\dd\O \cap \{x = (x',x_n) : |x'| < \rho\}$, 
then there are positive constants 
$\epsilon_0$, $C_0$, such that
\begin{equation}
\label{22}
\phi(y+z) + \phi(y-z) - 2\phi(y) \ge - C_0 |z|^2 \qquad 
\forall \ y \in \R^{n-1}: |y| <\rho, 
\quad \forall \ z \in \R^{n-1}: |z| < \epsilon_0 .
\end{equation}
\end{lemma}

\begin{proof}
Fix a point $y \in \R^{n-1}$, $|y| < \rho$.
Let $(p,q)$ be a center of a ball of radius $R = R(\O)$, 
touching the domain $\O$ from outside at the point $(y, \phi (y))$,
\begin{equation}
\label{23}
|y-p|^2 + (\phi(y)-q)^2 = R^2 .
\end{equation}
This external ball lies outside the cone 
$\{ (x', x_n)\in\R^n: x_n>\phi(y)+K|x'-y|\}$. 
Hence, $\frac{|y-p|}{\phi(y)-q} \le K$ and 
\begin{equation}
\label{24}
\phi(y)-q\ge\frac{R}{\sqrt{1+K^2}}  \,.
\end{equation}
Further,
$$
\phi(y+z) - q \ge \sqrt{R^2 - |y+z-p|^2} = 
\sqrt{R^2 - |y-p|^2} \left(1 - \frac{2 \<z,y-p\>}{R^2 - |y-p|^2}
- \frac{z^2}{R^2 - |y-p|^2}\right)^{1/2} .
$$
In view of \eqref{23} and the Taylor formula
$$
\sqrt{1+\al} = 1 + \frac\al{2} - \frac{\al^2}8 +O (\al^3), \quad \al \to 0,
$$
we obtain
\begin{equation}
\label{25}
\phi(y+z) - q \ge \phi(y) - q - \frac{\<z,y-p\>}{\phi(y) - q} 
- \frac{z^2}{2(\phi(y) - q)} - \frac{\<z,y-p\>^2}{2(\phi(y) - q)^3} 
+ O(|z|^3), \quad z \to 0 .
\end{equation}
Due to \eqref{24} we can choose 
$\epsilon > 0$ independent on $y$
such that
$$
\frac{\left|2 \<z,y-p\> +z^2\right|}{(\phi(y) - q)^2} \le \frac12 
\quad \text{for} \ |z| \le \epsilon .
$$
In this case the constant in the simbol $O$ in the formula \eqref{25} 
is uniform with respect to
$y \in \R^{n-1}$, $|y| < \rho$.
Summarizing the equation \eqref{25} for vectors $z$ and $-z$, we get
$$
\phi(y+z) + \phi(y-z) - 2\phi(y) \ge 
- \frac{z^2}{\phi(y) - q} - \frac{\<z,y-p\>^2}{(\phi(y) - q)^3} + O(|z|^3)
\ge - C_0 |z|^2 \quad \text{for} \ |z| \le \epsilon_0 
$$
for sufficiently small $\epsilon_0$.
\end{proof}

\begin{lemma}
\label{l22}
Let $\O\subset \R^n$ be a domain, $0\in\dd\O$,
$$
\hat{B}=\{(x', x_n)\in\R^n:|x'|^2+|x_n-R|^2<R^2\}\subset B_1(0),
$$
and $\hat{B}\cap \O=\emptyset$. 
Let $\psi:B_1(0)\to U$ be a $C^2$-diffeomorphism. 
Then the set $U$ contains a ball $\tilde{B}$ of radius $a$, 
nonintersecting with $\psi(\O\cap B_1(0))$. 
Also, $\psi(0)\in\dd\tilde{B}$, 
and the radius $a$ depends only on $R$ and on $\|\psi^{-1}\|_{C^2(U)}$.
\end{lemma}

\begin{proof}
Denote the new variables as $y=\psi(x)$, $x = \psi^{-1} (y)$.
Without loss of generality, one may assume that
\begin{equation}
\label{255}
x(0)=0,\quad 
\dfrac{\dd x_n}{\dd y_i}(0)=0,\ i = 1, \dots, n-1,
\quad \dfrac{\dd x_n}{\dd y_n}(0)>0.
\end{equation}
Let us show that for sufficiently small $a$ the implication
\begin{equation}
\label{256}
(y_n-a)^2+|y'|^2<a^2 \quad \Rightarrow \quad (x_n-R)^2+|x'|^2<R^2.
\end{equation}
holds.
Since $\psi^{-1}$ belongs to $C^2(U)$, 
and \eqref{255} holds, we have
$$
\quad |x|\le c_1|y|\quad \text{ and }\quad 
x_n\ge\gamma y_n - c_2|y|^2,\quad \gamma>0. 
$$
Substituting this estimates in \eqref{256}, we get 
that it is sufficient to establish an implication
$$
|y|^2<2a y_n \quad \Rightarrow \quad 
-2R\gamma y_n + (2Rc_2+c_1^2)|y|^2<0.
$$
It is true for $a<\dfrac{R\gamma}{2Rc_2+c_1^2}$.
\end{proof}

{\it Proof of the Theorem \ref{t01}.}
Let $\O$ be a Lipschitz domain, satisfying EBC, $x_0\in\dd\O$.
Let $D$ and $\Lambda$ be a neighbourhood of the point $x_0$ 
and a special Lipschitz domain from the Definition \ref{o02} respectively.
Without loss of generality, we may suppose that $x_0 = 0$, $D = B_{2\rho} (0)$. 
By Lemma \ref{l21} the function $\phi$ defining $\Lambda$ 
satisfies \eqref{22} with some constants $C_0$ and $\epsilon_0$.
Consider $C^2$-diffeomorphism $\psi: \R^n \to \R^n$,
$$
\psi: (x',x_n) \mapsto (y',y_n) = (x',x_n + \xi(x')) ,
$$ 
where $\xi$ is a smooth function of $x'$, 
$\xi(x') = C_0 |x'|^2$ for $|x'| < \rho$, 
and $\xi$ satisfies also the Lipschitz condition on the whole $\R^{n-1}$.
Domain $\psi(\Lambda)$ is the special Lipschitz domain,
which is determined by function $\tilde{\phi}$, 
satisfying inequality
$$
\tilde{\phi}(y'+z) + \tilde{\phi}(y'-z) - 2\tilde{\phi}(y') \ge 0 \qquad 
\forall \ y' \in \R^{n-1}: |y'| <\rho, 
\quad \forall \ z \in \R^{n-1}: |z| < \epsilon_0 .
$$
Thus, the set $\psi(\Lambda) \cap B_r(0)$, $r<\rho$ is convex; 
one can take $\psi(\Lambda)$ and 
$\psi^{-1} (B_r (0))$ with sufficiently small $r$ as 
the neighbourhoods $V$ and $U$ from the Definition \ref{o03}.
Therefore $\O \in {\cal C} (C^2)$. 

Let us establish the inverse inclusion. 
Since convex domains satisfy EBC, it follows from Lemma \ref{l22} 
that at every point of the boundary of a domain of class $\mathcal{C}(C^2)$ 
there is an external ball, and its radius can be chosen 
to be independent of the point. 
$\qed$

\begin{rem}
Actually we have shown that the class of domains with Lipschitz boundary,
satisfying EBC, coincides with $\mathcal{C}(C^\infty)$ too.
\end{rem}

\subsection{Approximation of the convex domains by the smooth ones}

\begin{theorem}
\label{t22}
Let $\O \subset \R^n$ be a special Lipschitz domain \eqref{20}, \eqref{21}, 
and the set
$$
\O \cap \{ x = (x',x_n) : |x'| < 2\rho \}
$$ 
be convex.
Then there exists an ascending collection of special Lipschitz domains
$\{\O_\al\}_{\al \in (0,1)}$ of class $C^\infty$, such that

1) The domains $\O_\al$ are described by the formula
\begin{equation}
\label{26}
\O_\al = \{ (x', x_n) \in \R^n : x_n > \phi_\al (x') \} ,
\end{equation}
\begin{equation*}
|\phi_\al (x_1') - \phi_\al (x_2')| \le K |x'_1-x'_2| \ \ \forall\ x_1', x_2' ,
\end{equation*}
$\mathop{\bigcup}\limits_{\al \in (0,1)} \O_\al = \O$,
and the domains
$\O_\al \cap \{ x = (x',x_n) : |x'| < \rho \}$ are convex.

2) There are extension operators
$$
P_\al : W_2^1 (\O_\al) \to W_2^1 (\R^n), \quad  
\|P_\al\| \le C_1 ,
$$ 
such that
$\|\n (P_\al u)\|_{L_2 (\R^n)} \le C_1 \|\n u\|_{L_2 (\O_\al)}$
for all  $u \in W_2^1 (\O_\al)$.

3) For $n>2$ 
$$
\|u\|_{L_{\frac{2n}{n-2}} (\O_\al)} \le C_2 \|\n u\|_{L_2(\O_\al)} 
\quad \forall \ u \in W_2^1 (\O_\al) .
$$

The constants $C_1, C_2$ do not depend on $\al$. 
\end{theorem}

\begin{rem}
Theorem \ref{t22} is taken from \cite{Mi}, where its analogue in 
more general case of special Lipschitz domains satisfying EBC was obtained. 
The case of convex domains is sufficient for our purposes.
To give a complete picture, we provide the corresponding proof.
\end{rem}

\begin{proof}
Take 
$$
\oo \in C_0^\infty (\R^{n-1}), 
\quad \supp \oo \subset B_\rho (0),
\quad \oo(x') \ge 0, 
\quad \int_{\R^{n-1}} \oo(x')\, dx' = 1.
$$
Define functions $\phi_\al$ as mollifications of functions $\phi$, shifted by a constant:
\begin{equation}
\label{27}
\phi_\al (x') = M \al + 
\int_{\R^{n-1}} \oo(z') \ph (x'-\al z')\, dz' ,
\end{equation}
where 
$M > K \int_{\R^{n-1}} \oo(z') |z'|\, dz'$.
Then
\begin{equation}
\label{28}
\frac{\dd\phi_\al (x')}{\dd\al} = 
M - \int_{\R^{n-1}} \oo(z') \<z', \n\ph (x'-\al z')\> dz' \ge 
M - K \int_{\R^{n-1}} \oo(z') |z'|\, dz' > 0 .
\end{equation}
We define the domains $\O_\al$ by the formula \eqref{26}.

1) It is clear that
\begin{equation}
\label{29}
\left|\phi_\al(x_1') - \phi_\al (x_2')\right| =
\left|\int_{\R^{n-1}} \oo(z') 
\left(\phi (x_1'-\al z') - \phi (x_2'-\al z')\right)dz' \right| 
\le K |x'_1-x'_2| .
\end{equation}
It follows from definition \eqref{27} that
\begin{eqnarray*}
\phi_\al (x'+y') + \phi_\al (x'-y') - 2\phi_\al (x') \\
= \int_{\R^{n-1}} \oo(z') 
\left(\phi(x'+y'-\al z') + \phi(x'-y'-\al z') - 2\phi(x'-\al z')\right) dz' 
\ge 0 
\end{eqnarray*}
for all $x', y'$, such that $x' \pm y' \in B_\rho$,
i.e. functions $\phi_\al$ are convex in $B_\rho$.

By \eqref{28} and convergence
$\phi_\al (x') \mathop{\longrightarrow}\limits_{\al\to 0} \phi (x')$ we get 
that the domains $\O_\al$ are ascending and together they cover the whole $\O$,
$\mathop{\bigcup}\limits_{\al \in (0,1)} \O_\al = \O$.

2) It is well known (see \cite{St}) that for special Lipschitz domains
there exists an extension operator
$P_\al : W_2^1 (\O_\al) \to W_2^1 (\R^n)$,
$\|P_\al\| \le C_1$,
and the estimate 
$\|\n (P_\al u)\|_{L_2 (\R^n)} \le C_1 \|\n u\|_{L_2 (\O_\al)}$ 
holds with constant $C_1$ depending only on the Lipschitz constant 
of the domain $\O_\al$.
It remains to refer to \eqref{29}.

3) By the Sobolev inequality we have
$$
\|u\|_{L_{\frac{2n}{n-2}} (\O_\al)} \le 
\|P_\al u\|_{L_{\frac{2n}{n-2}} (\R^n)} \le 
C \|\n (P_\al u)\|_{L_2 (\R^n)} \le C C_1 \|\n u\|_{L_2 (\O_\al)}
\quad \forall \ u \in W_2^1 (\O_\al) . \quad \qedhere
$$
\end{proof}
%
\section{Auxiliary statements}
\subsection{Algebraic lemma}

\begin{lemma}
\label{l31}
Let $B$ be a self-adjoint $(3\times 3)$ matrix,
$$
0 < \be_0 \1 \le B \le \be_1 \1 \,,
$$ 
$U$ be an arbitrary $(3\times 3)$ matrix.
Then
\begin{equation}
\label{31}
\tr (B U \overline{B U}) + 
\be_1^2 \left( \tr (U U^*) - \tr (U \overline U) \right) \ge 
\be_0^2 \tr (U U^*) .
\end{equation}
\end{lemma}

\begin{proof}
The matrix $B$ can be diagonalized, $B = O\Lambda O^*$, 
where $O$ is the unitary matrix,
$$
\L = \operatorname{diag} (\la_1, \la_2, \la_3), \qquad 
\be_1 \ge \la_1 \ge \la_2 \ge \la_3 \ge \be_0 .
$$
Put $W=O^* U \overline{O}$. 
Then \eqref{31} is equivalent to the inequality
$$
\tr(\Lambda W \overline{\Lambda W}) +
\be_1^2 \left(\tr(WW^*)-\tr(W\overline{W})\right) 
\ge \be_0^2 \tr(WW^*)
$$
or
$$
\sum_{i,j} \left( \la_i \la_j w_{ij} \overline{w_{ji}} 
+ \be_1^2 (w_{ij} \overline{w_{ij}} - w_{ij} \overline{w_{ji}}) \right)
\ge \be_0^2 \sum_{i,j} w_{ij} \overline{w_{ij}} .
$$
Since $\la_i \ge \be_0$, it is sufficient to consider the terms with $i \neq j$.
The corresponding inequality "splits" 
into three independent inequalities:
\begin{equation}
\label{32}
\la_1 \la_2 (w_{12} \overline{w_{21}} + w_{21} \overline{w_{12}}) +
\be_1^2 \left(|w_{12}|^2 + |w_{21}|^2 - 
w_{12} \overline{w_{21}} - w_{21} \overline{w_{12}}\right) \ge 
\be_0^2 \left(|w_{12}|^2 + |w_{21}|^2\right)
\end{equation}
and the same inequality for pairs of indices $\{1,3\}$ and $\{2,3\}$.
Clearly, the left hand side of \eqref{32} 
can be estimated from below in the following way:
\begin{eqnarray*}
\be_1^2 \left(|w_{12}|^2 + |w_{21}|^2\right) - 
2 (\be_1^2 - \la_1 \la_2) \re (w_{12} \overline{w_{21}})\\
\ge \be_1^2 \left(|w_{12}|^2 + |w_{21}|^2\right) - 
(\be_1^2 - \la_1 \la_2) \left(|w_{12}|^2 + |w_{21}|^2\right) 
\ge \be_0^2 \left(|w_{12}|^2 + |w_{21}|^2\right) . \quad \qedhere
\end{eqnarray*}
\end{proof}

\subsection{Estimates of minor terms}
\begin{lemma}
\label{l32}
Let $\O \subset \R^3$ be a special Lipschitz domain,
the set $\O \cap \{ x = (x',x_n) : |x'| < 2\rho \}$ be convex,
$0\in \dd\O$.
Let $v$ be a vector-function with compact support, $v\in W_2^1 (\O, \R^3)$, 
$\supp v \subset \overline\O \cap B_\rho (0)$.
Let $s$ be a self-adjoint $(3\times 3)$-matrix-function,
$s \in W_3^1 (\O \cap B_\rho)$ and 
$$
0 < \be_0 \1 \le s \le \be_1 \1 \,.
$$ 
Take $\de > 0$.
Then 
\begin{equation}
\label{33}
\int_{\O_\al\cap B_\rho} |\dd_i s_{jk}|^2 |v_l|^2 dx \le 
\de \|\n v\|_{L_2 (\O_\al)}^2 + C(\de, \rho, s) \|v\|_{L_2 (\O_\al)}^2
\end{equation}
and 
\begin{equation}
\label{34}
\int_{\O_\al\cap B_\rho} |s_{ij} \dd_k s_{lm} v_n \dd_p v_q|\, dx 
\le \de \|\n v\|_{L_2 (\O_\al)}^2 + C(\de, \rho, s) \|v\|_{L_2 (\O_\al)}^2 ,
\end{equation}
where $\{\O_\al\}$ is the collection of domains, 
constructed in Theorem \ref{t22};
the constant $C(\de, \rho, s)$ is independent on $\al$ and $v$.
\end{lemma}

\begin{rem}
These estimates are well known. 
To give a complete picture, we provide the proof.
\end{rem}

\begin{proof}
For any $\de_0 > 0$, the function $\dd_i s_{jk} \in L_3 (\O \cap B_\rho)$
can be represented in the form
$\dd_i s_{jk} = \nu_1 + \nu_2$, where $\nu_1 \in L_\infty (\O \cap B_\rho)$,
$\nu_2 \in L_3 (\O \cap B_\rho)$, and
$\|\nu_2\|_{L_3 (\O \cap B_\rho)} \le \de_0$.
Therefore,
\begin{equation}
\label{35}
\|\dd_i s_{jk} v_l\|_{L_2 (\O_\al)} \le 
\|\nu_1\|_{L_\infty (\O \cap B_\rho)} \|v_l\|_{L_2 (\O_\al)} +
\de_0 \|v_l\|_{L_6 (\O_\al)} \le 
C \de_0 \|\n v_l\|_{L_2 (\O_\al)} + C (\de_0, \rho, s) \|v_l\|_{L_2 (\O_\al)} ,
\end{equation}
where we used the statement 3) of the Theorem \ref{t22}.
Thus, \eqref{33} is proved.
Furthermore, in view of \eqref{35} we have the estimate
\begin{eqnarray*}
\int_{\O_\al\cap B_\rho} |s_{ij} \dd_k s_{lm} v_n \dd_p v_q|\, dx 
\le \be_1 \|\dd_k s_{lm} v_n\|_{L_2 (\O_\al)} \|\dd_p v_q\|_{L_2 (\O_\al)} \\
\le \be_1 \left(C \de_0 \|\n v_n\|_{L_2 (\O_\al)} 
+ C (\de_0, \rho, s) \|v_n\|_{L_2 (\O_\al)}\right) \|\n v_q\|_{L_2 (\O_\al)} ,
\end{eqnarray*}
which implies \eqref{34}.
\end{proof}

From lemmas \ref{l31} and \ref{l32} one can deduce

\begin{lemma}
\label{l33}
Let $\O \subset \R^3$ be a special Lipschitz domain, the set 
$$
\O \cap \{ x = (x',x_n) : |x'| < 2\rho \}
$$ 
is convex, $0\in \dd\O$.
Let $v$ be a vector-function with compact support, 
$$
v\in W_2^1 (\O, \R^3), \quad \supp v \subset \overline\O \cap B_\rho (0) .
$$
Let $s$ be a self-adjoint $(3\times 3)$ matrix-function,
$s \in W_3^1 (\O \cap B_\rho)$ and
$$
0 < \be_0 \1 \le s \le \be_1 \1 \,.
$$ 
Then 
$$
\int_{\O_\al\cap B_\rho} 
\left(|\n (sv)|^2 - |\rot (sv)|^2 + \be_1^2 |\rot v|^2\right) dx
\ge \frac{\be_0^2}2 \int_{\O_\al} |\n v|^2 dx - 
C(\rho, s) \int_{\O_\al} |v|^2 dx,
$$
where $\{\O_\al\}$ is the collection of the domains, 
constructed in Theorem \ref{t22};
the constant $C(\rho, s)$ is independent on $\al$ and $v$.
\end{lemma}

\begin{proof}
We have
$$
|\n v|^2 = \tr (U U^*) , \quad 
|\n v|^2 - |\rot v|^2 = \tr (U \overline U), \quad 
\text{where} \ U_{kj}= \dd_j v_k .
$$
Therefore,
$$
|\n (sv)|^2 - |\rot (sv)|^2 + \be_1^2 |\rot v|^2 = 
\tr (\tilde U \overline{\tilde U}) +
\be_1^2 \left( \tr (U U^*) - \tr (U \overline U) \right) ,
$$
where $\tilde U_{ij} = \dd_j (s_{ik} v_k)$.
So, the main terms on the left hand side (when all derivatives fall on $v$) 
are estimated by Lemma \ref{l31}, applied to matrices 
$B=s$, $U_{kj}= \dd_j v_k$,
\begin{eqnarray*}
\int_{\O_\al\cap B_\rho} \left(\tr (B U \overline{B U}) + 
\be_1^2 \left( \tr (U U^*) - \tr (U \overline U) \right) \right) dx \ge 
\be_0^2 \int_{\O_\al} |\n v|^2 dx .
\end{eqnarray*}
The minor terms do not exceed
$\frac{\be_0^2}2 \|\n v\|_{L_2(\O_\al)}^2 
+ C(\rho, s) \|v\|_{L_2(\O_\al)}^2$ 
by virtue of Lemma \ref{l32}.
\end{proof}

\subsection{Density of smooth functions}
We will need to approximate functions from Sobolev space, 
satisfying tangent or normal boun\-da\-ry condition, 
by smooth functions with the same boundary condition.
We give the proof of this fact for the convenience of the reader.

\begin{lemma}
\label{l335}
Let $\O\subset \R^3$ be a special Lipschitz domain, 
\begin{equation*}
\O = \{x \in \R^3 : x_3 > \phi (x_1,x_2) \} ,
\end{equation*}
and $\phi \in C^3 (\R^2)$, $\phi (0) = 0$.
Then

a) for every $u\in W_2^1 (\O, \tau)$, $\supp u \subset B_\rho$, 
there exists a sequence of functions \\
$u^k\in W_2^1 (\O, \tau) \cap C^2(\overline\O, \R^3),$ 
such that $\supp u^k \subset B_\rho$, $u^k\mathop{\longrightarrow}\limits_{k\rightarrow\infty}^{W_2^1}u$;

b) for every $u\in W_2^1 (\O,\1, \nu)$, $\supp u \subset B_\rho$, there exists a sequence of functions \\
$u^k\in W_2^1 (\O,\1, \nu) \cap C^2(\overline\O, \R^3),$ 
such that $\supp u^k \subset B_\rho$, $u^k\mathop{\longrightarrow}\limits_{k\rightarrow\infty}^{W_2^1}u$.
\end{lemma}

\begin{proof} 
Let us consider the matrix-function
$$
M(x) = \left(\begin{array}{ccc} 
1&0&\dd_1 \phi (x_1, x_2) \\ 
0&1& \dd_{2} \phi (x_1, x_2) \\
-\dd_1 \phi (x_1, x_2) &-\dd_2 \phi (x_1, x_2) & 1 \end{array}\right) .
$$
This matrix is nondegenerate, $\det M=1+|\n\phi|^2\neq 0$.
Take $u\in W_2^1 (\O,\C^3),$ $\supp u \subset B_\rho$. 
Denote $v=Mu\in W_2^1 (\O,\C^3),$ $\supp v \subset B_\rho$. 
Then
$$
\left.u_\tau\right|_{\dd\O}=0 \Leftrightarrow \left.v_1\right|_{\dd\O}=\left.v_2\right|_{\dd\O}=0,
$$ 
$$
\left.u_\nu\right|_{\dd\O}=0 \Leftrightarrow \left.v_3\right|_{\dd\O}=0.
$$ 
We will consider electric case. 
Magnetic case is treated in a similar way. 
We can approximate $v_1, v_2$ by
$v_1^k, v_2^k \in C_0^\infty(B_\rho)$, and 
$v_3$ by $v_3^k\in C^2(\overline\O)$, $\supp v \subset B_\rho$. 
Then
$$
u^k=M^{-1}v^k\in W_2^1 (\O, \tau) \cap C^2(\overline\O, \R^3), \quad u^k\mathop{\longrightarrow}\limits_{k\rightarrow\infty}^{W_2^1}u.\quad
\qedhere$$ 
\end{proof}

\section{A priori estimates}
\subsection{Magnetic fields}
\begin{lemma}
\label{l34}
Let $\O\subset \R^3$ be a special Lipschitz domain, 
\begin{equation*}
\O = \{x \in \R^3 : x_3 > \phi (x_1,x_2) \} ,
\end{equation*}
and $\phi \in C^3 (\R^2)$.
Let $w \in W_2^1 (\O)$, $\left. w_\nu\right|_{\dd\O} = 0$.
Then the integration by part formula
$$
\int_\O \left( |\rot w|^2 + |\div w|^2 \right) dx = 
\int_\O |\n w|^2 dx + \int_{\dd\O} \< Aw, w \> dS 
$$
holds, where
\begin{equation}
\label{36}
A (x_1, x_2, \phi(x_1, x_2)) = \frac1{\sqrt{1+|\n\phi(x)|^2}} 
\left( \begin{array}{ccc}
\dd_1^2 \phi (x) & \dd_1 \dd_2 \phi (x) & 0 \\
\dd_1 \dd_2 \phi (x) & \dd_2^2 \phi (x) & 0 \\
0 & 0 & 0 \end{array} \right) 
\end{equation}
is the Weingarten mapping (see \cite{BZ}).
\end{lemma}

\begin{proof}
In view of Lemma \ref{l335} it is sufficient to consider smooth functions $w$.
We have
\begin{eqnarray}
\nonumber
\int_\O \left( |\n w|^2 - |\rot w|^2 \right) dx = 
\int_\O \dd_j w_k \dd_k \overline w_j \, dx = 
- \int_\O w_k \dd_j \dd_k \overline w_j \, dx + 
\int_{\dd\O} \nu_j w_k \dd_k \overline w_j \, dS \\
= \int_\O |\div w|^2 dx + \int_{\dd\O} \nu_j w_k \dd_k \overline w_j \, dS
- \int_{\dd\O} \nu_k w_k \dd_j \overline w_j \, dS ,
\label{37}
\end{eqnarray}
where
$$
\nu (x_1, x_2, \phi(x_1, x_2)) = \frac1{\sqrt{1+|\n\phi(x)|^2}} 
\left( \begin{array}{cc} 
\dd_1 \phi (x) \\ \dd_2 \phi (x) \\ -1 
\end{array}\right) 
$$
is the unit external normal to the boundary $\dd\O$.
We assume $\left. \<w,\nu\> \right|_{\dd\O} = 0$,
therefore the last term in \eqref{37} vanishes. 
Moreover, $w_k \dd_k (\nu_j \overline w_j) = 0$, 
since the operator $w_k \dd_k$ acts in the tangent plane only.
Thus,
$$
\int_\O \left( |\rot w|^2 + |\div w|^2 - |\n w|^2 \right) dx = 
\int_{\dd\O} \dd_k \nu_j w_k \overline w_j \, dS .
$$
Furthermore, 
$$
\dd_k \nu =  \frac1{\sqrt{1+|\n\phi|^2}} \ 
\dd_k \left( \begin{array}{cc} 
\dd_1 \phi \\ \dd_2 \phi \\ -1 
\end{array}\right) + 
\dd_k \left(\frac1{\sqrt{1+|\n\phi|^2}}\right) 
\left( \begin{array}{cc} 
\dd_1 \phi \\ \dd_2 \phi \\ -1 
\end{array}\right) .
$$
Using again the condition $\left. \<w,\nu\> \right|_{\dd\O} = 0$,
one can deduce from here that
$$
\<\overline w, \dd_k \nu\> = 
\frac1{\sqrt{1+|\n\phi|^2}} \left\< \overline w,
\left( \begin{array}{cc} 
\dd_1 \dd_k \phi \\ \dd_2 \dd_k \phi \\ 0 
\end{array}\right) \right\>
$$
and
$w_k \overline w_j \dd_k \nu_j = \< Aw, w \>$.
\end{proof}

\begin{theorem}
\label{t35}
Let $\O$ be a special Lipschitz domain, 
$$
\O = \{x \in \R^3 : x_3 > \phi (x_1,x_2) \},
$$
the set $\O \cap \{x \in \R^3 : \sqrt{x_1^2 + x_2^2} <2 \rho \}$ be convex,
$\O_\al = \{x \in \R^3 : x_3 > \phi_\al (x_1,x_2) \}$
be the collection of domains constructed in Theorem \ref{t22}.
Let $0 \in \dd\O$, matrix-function $\mu$ be defined in $\O \cap B_\rho (0)$
and satisfy \eqref{01} and \eqref{04} in $\O\cap B_\rho (0)$.
Then
$$
\|v\|_{W_2^1 (\O_\al)} \le C(\rho, \mu) \|v\|_{F(\O_\al, \mu)} 
\quad \forall v \in W_2^1 (\O_\al, \mu, \nu), \ \supp v \subset B_\rho,
$$
the constant $C(\rho, \mu)$ does not depend on $v$ and $\al$.
\end{theorem}

\begin{proof}
Take $v \in W_2^1 (\O_\al, \mu, \nu)$, $\supp v \subset B_\rho$.
By virtue of Lemma \ref{l34}, 
applied to the domain $\O_\al$ and the function $w = \mu v$, we have
\begin{eqnarray*}
\int_{\O_\al} |\div (\mu v)|^2 dx + \mu_1^2 \int_{\O_\al} |\rot v|^2 dx \\
= \int_{\O_\al} \left(|\n (\mu v)|^2 - |\rot (\mu v)|^2 + 
\mu_1^2 |\rot v|^2\right) dx
 + \int_{\dd\O_\al} \< A_\al \mu v, \mu v \> dS =: I_1 + I_2,
\end{eqnarray*}
where the matrix $A_\al$ is determined by the formula \eqref{36} 
with the function $\phi$ changed by $\phi_\al$.
By Lemma \ref{l33} 
$$
I_1 \ge \frac{\mu_0^2}2 \int_{\O_\al} |\n v|^2 dx - 
C(\rho, \mu) \int_{\O_\al} |v|^2 dx.
$$
Furthermore, by Theorem \ref{t22} $D^2\phi_\al\ge 0$
for $\sqrt{x_1^2 + x_2^2} < \rho$, therefore $I_2\ge 0$. 
Thus, 
$$
\int_{\O_\al} |\div (\mu v)|^2 dx + \mu_1^2 \int_{\O_\al} |\rot v|^2 dx 
\ge \frac{\mu_0^2}2 \int_{\O_\al} |\n v|^2 dx - 
C(\rho, \mu) \int_{\O_\al} |v|^2 dx.\quad
\qedhere
$$
\end{proof}

\subsection{Electric field}

Recall that we assume the coefficient $s$ to be real.
Note that it is the subsection where we use this assumption. 

\begin{lemma}
\label{l36}
Let $\O \subset \R^3$ be a bounded domain, $\dd\O \in C^2$.
Let $u \in C^2 (\overline\O, \R^3)$, 
$s$ be a $(3\times 3)$ matrix-function with real entries, 
$s \in W_3^1 (\O)$.
Then
$$
\int_\O \left(|\rot(su)|^2 + |\div(su)|^2\right) dx = 
\int_\O |\n (su)|^2 dx + \int_{\dd\O} K(x, u(x)) \, dS + I_3 ,
$$
where 
\begin{equation}
\label{415}
K(x, u(x)) = s_{jm} s_{kn} (\nu_k \dd_j u_m - \nu_j \dd_k u_m) \overline u_n ,
\end{equation}
$\nu(x)$ is the unit external normal to $\dd\O$,
$I_3$ is the linear combination of integrals of type
\begin{equation}
\label{38}
\int_\O \dd_i s_{jk} \dd_l s_{mn} u_p \overline u_q dx 
\quad \text{and} \quad 
\int_\O s_{ij} \dd_k s_{lm} u_n \dd_p \overline u_q dx .
\end{equation}
\end{lemma}

\begin{rem}
In this proof we will denote different linear combinations of type \eqref{38} 
by the same letter $I_3$.
\end{rem}

\begin{proof}
We have
\begin{eqnarray*}
\int_\O \left(|\n (su)|^2 - |\rot(su)|^2\right) dx = 
\int_\O \dd_k (s_{jm} u_m) \dd_j (s_{kn} \overline u_n) dx \\
= \int_\O s_{jm} \dd_k u_m s_{kn} \dd_j \overline u_n dx + I_3 =
- \int_\O s_{jm} \dd_j \dd_k u_m s_{kn} \overline u_n dx + 
\int_{\dd\O} s_{jm} \dd_k u_m s_{kn} \nu_j \overline u_n dS + I_3 \\
= \int_\O s_{jm} \dd_j u_m \dd_k (s_{kn} \overline u_n) dx + 
\int_{\dd\O} s_{jm} \dd_k u_m s_{kn} \nu_j \overline u_n dS 
- \int_{\dd\O} s_{jm} \nu_k \dd_j u_m s_{kn} \overline u_n dS + I_3 \\
= \int_\O |\div(su)|^2 dx - \int_{\dd\O} K(x, u(x)) \, dS + I_3 .\quad
\qedhere
\end{eqnarray*}
\end{proof}

\begin{lemma}
\label{l37}
Let $\dd\O \in C^2$, $0 \in \dd\O$. 
Assume that in the neighbourhood of zero
domain $\O$ is described by the formula $x_3 > \psi (x_1, x_2)$, and
$$
\psi (0) = 0, \quad \n \psi (0) = 0, \quad D^2 \psi (0) \ge 0 .
$$
Let $s$ be a fixed matrix with real entries, $s > 0$.
Let $u \in C^1 (\overline \O, \R^3)$, $\left. u_\tau \right|_{\dd\O} = 0$.
Then
$$
K (0, u(0)) = s_{jm} s_{kn} 
(\nu_k (0) \dd_j u_m (0) - \nu_j (0) \dd_k u_m (0)) \overline u_n (0) 
\ge 0 .
$$
\end{lemma}

\begin{proof}
We have 
$$
\nu (x) = \frac1{\sqrt{1+|\n\psi(x)|^2}} 
\left( \begin{array}{cc} 
\dd_1 \psi (x) \\ \dd_2 \psi (x) \\ -1 
\end{array}\right) ,
$$
and in particular,
\begin{equation}
\label{39}
\nu (0) = \left( \begin{array}{cc} 
0 \\ 0 \\ -1 \end{array}\right) .
\end{equation}
If we introduce the tangent vectors
$$
\tau_1 (x) = \left( \begin{array}{cc} 
1 \\ 0 \\ \dd_1 \psi (x) \end{array}\right) , \qquad
\tau_2 (x) = \left( \begin{array}{cc} 
0 \\ 1 \\ \dd_2 \psi (x) \end{array}\right) ,
$$
then the condition $\left. u_\tau \right|_{\dd\O} = 0$
can be rewritten in the form
\begin{equation}
\label{310}
\< u(x_1, x_2, \psi (x_1, x_2)), \tau_j (x_1, x_2)\> = 0, \quad j = 1, 2 ;
\end{equation}
At the origin it becomes
\begin{equation}
\label{311}
u_1(0) = u_2(0) = 0 .
\end{equation}
Differentiating \eqref{310} with respect to $x_i$, $i = 1, 2$, at the point $0$,
and taking into account the equality $\n\psi (0) = 0$, we get
\begin{equation}
\label{312}
\dd_i u_j (0)= -u_3 (0) \dd_i \dd_j \psi (0), \quad i, j = 1, 2 .
\end{equation} 
According to \eqref{39} and \eqref{311} 
$$
K (0, u(0)) = - s_{jm} s_{33} \dd_j u_m (0) \overline u_3 (0) 
+ s_{3m} s_{k3} \dd_k u_m (0) \overline u_3 (0) .
$$
The terms with $m=3$ cancel out.
Moreover, the first term with $j=3$ cancels with the second term with $k=3$.
Therefore,
$$
K (0, u(0)) = \sum_{j,m = 1,2} 
\left(s_{jm} s_{33} - s_{3m} s_{j3}\right) \dd_j \dd_m \psi (0) |u(0)|^2 ,
$$
where we used \eqref{312}.
Without loss of generality, one may assume 
that matrix $(\dd_j \dd_m \psi (0))$ is diagonal,
$$
(D^2 \psi (0)) = 
\left( \begin{array}{cc}
\la_1 & 0 \\
0 & \la_2 \end{array} \right),
\quad \la_i \ge 0 .
$$
Therefore,
$$
K (0, u(0)) = \sum_{j = 1,2} 
\left(s_{jj} s_{33} - s_{3j} s_{j3}\right) \la_j |u(0)|^2 \ge 0, 
$$
because the matrix $s$ is positively definite.
\quad$\qedhere$\end{proof}

\begin{theorem}
\label{t38}
Let $\O$ be a special Lipschitz domain, $0 \in \dd\O$,
the set
$$
\O \cap \{x \in \R^3 : \sqrt{x_1^2 + x_2^2} <2 \rho \}
$$ 
be convex, $\O_\al = \{x \in \R^3 : x_3 > \phi_\al (x_1,x_2) \}$
be domains constructed in Theorem \ref{t22} and filling $\O$.
Let $\er$ be a matrix-function defined in $\O\cap B_\rho (0)$ 
and satisfying \eqref{01} and \eqref{04} in $\O\cap B_\rho (0)$.
Then
$$
\|u\|_{W_2^1 (\O_\al)} \le C(\rho, \er) \|u\|_{F(\O_\al, \er)} 
\quad \forall u \in W_2^1 (\O_\al, \tau), \ \supp u \subset B_\rho,
$$
and the constant $C(\rho, \er)$ does not depend on $v$ and $\al$.
\end{theorem}

\begin{proof}
Due to Lemma \ref{l335} it is sufficient to consider smooth functions $u$.
By virtue of Lemma \ref{l36} 
\begin{eqnarray*}
\int_{\O_\al} |\div (\er u)|^2 dx + \er_1^2 \int_{\O_\al} |\rot u|^2 dx \\
= \int_{\O_\al} \left(|\n (\er u)|^2 - |\rot (\er u)|^2 
+ \er_1^2 |\rot u|^2\right) dx
 + \int_{\dd\O_\al} K(x, u(x))\, dS + I_3,
\end{eqnarray*}
where $K(x, u(x))$ is defined by \eqref{415}, 
$I_3$ is a linear combination of integrals of type \eqref{38}.
By Lemma \ref{l33} 
$$
\int_{\O_\al} \left(|\n (\er u)|^2 - |\rot (\er u)|^2 + \er_1^2 |\rot u|^2\right) dx
\ge \frac{\er_0^2}2 \int_{\O_\al} |\n u|^2 dx 
- C(\rho, \er) \int_{\O_\al} |u|^2 dx .
$$
By construction of $\O_\al$,  
$(D^2 \phi_\al)(x) \ge 0$ for $\sqrt{x_1^2 + x_2^2} <\rho$.
The trace of a positive definite matrix-function on the boundary 
is also positive definite, so $K (x, u(x)) \ge 0$ for a.e. $x \in \dd\O_\al$
due to the Lemma \ref{l37}.
Therefore,
$$
\int_{\dd\O_\al} K(x, u(x))\, dS \ge 0.
$$
Finally, by Lemma \ref{l32} 
$$
|I_3| \le \frac{\er_0^2}4 \int_{\O_\al} |\n u|^2 dx 
+ C(\rho, \er) \int_{\O_\al} |u|^2 dx,
$$
and therefore 
$$
\int_{\O_\al} \left(|\div (\er u)|^2 + \er_1^2 |\rot u|^2\right) dx 
\ge \frac{\er_0^2}4 \int_{\O_\al} |\n u|^2 dx 
- C(\rho, \er) \int_{\O_\al} |u|^2 dx .\quad
\qedhere
$$
\end{proof}

\section{Proof of Theorem \ref{t11}}
\subsection{Case of special Lipschitz domain}

\begin{lemma}
\label{l51}
Let $\{\O_k\}$ be an ascending sequence of domains, 
$\O_k \subset \O_{k+1}$, $\cup_k \O_k = \O$.
Suppose that $f_k \in L_2 (\O_k)$, $\|f_k\|_{L_2 (\O_k)} \le C_0$, 
and there exists $f \in L_2 (\O)$, such that 
$\left.f_k\right|_{\O_m} \to \left.f\right|_{\O_m}$ 
weakly in $L_2 (\O_m)$
\footnote{Here and in similar situations below we mean
tendency via the sequence $k = m, m+1, \dots$.}
for all $m$.
Then
$$
(f_k, z)_{L_2 (\O_k)} \to (f, z)_{L_2(\O)} \quad \forall \ z \in L_2 (\O) .
$$
\end{lemma}

\begin{proof}
Let $\er>0$. 
Choose $m$ such that 
$$
(C_0 + \Vert f\Vert_{L_2 (\O)}) \, \Vert z\Vert_{L_2 (\O\setminus\O_m)}\le \er/2.
$$
Then for $k>m$ we have 
\begin{eqnarray*}
|(f_k, z)_{L_2 (\O_k)}-(f, z)_{L_2 (\O)}| \\
\le |(f_k, z)_{L_2 (\O_m)}-(f, z)_{L_2 (\O_m)}|+ 
(\Vert f_k\Vert_{L_2 (\O_k)}  
+\Vert f\Vert_{L_2 (\O)}) \Vert z\Vert_{L_2 (\O\setminus\O_m)} \\
\le |(f_k, z)_{L_2 (\O_m)}-(f, z)_{L_2 (\O_m)}| +\er/2.
\end{eqnarray*}
It remains to choose $k$ such that the first term is less than $\er/2$.
\end{proof}

Recall that the operators ${\cal L}$ and ${\cal L}_s$
are determined by the formula \eqref{08} on the domains
\begin{equation*}
{\cal D} (\O) := F(\O, \er, \tau) \oplus W_2^1 (\O, \C) 
\oplus F (\O, \mu, \nu) \oplus \mathring W_2^1 (\O, \C) 
\end{equation*}
and
$$
{\cal A} (\O) := W_2^1 (\O, \tau) \oplus W_2^1 (\O, \C) 
\oplus W_2^1 (\O, \mu, \nu) \oplus \mathring W_2^1 (\O, \C) .
$$
It is natural to introduce the graph norm
$\|X\|_{{\cal D} (\O)}^2 = \|{\cal L} X\|^2_{L_2} + \|X\|_{L_2}^2$
as a norm in the space ${\cal D} (\O)$,
and the $W_2^1$-norm as a norm in the space ${\cal A} (\O)$.

\begin{theorem}
\label{t52}
Let $\O\subset \R^3$ be a special Lipschitz domain, the set
$$
\O \cap \{x \in \R^3 : \sqrt{x_1^2 + x_2^2} <2 \rho \}
$$ 
be convex, $0 \in \dd\O$.
Let $\er$ and $\mu$ be matrix-functions defined in $\O\cap B_{2\rho} (0)$ 
and satisfying there conditions \eqref{01} and \eqref{04}.
Let
$$
X = \left( \begin{array}{cc} E \\ \ph \\ H \\ \eta\end{array} \right) 
\in {\cal D} (\O), \qquad \supp X \subset \overline\O \cap B_{\rho/2} (0) .
$$
Then 
$$
X \in {\cal A} (\O) , \quad 
\|X\|_{W_2^1} \le C(\rho,\er,\mu) (\|{\cal L} X\|_{L_2} + \|X\|_{L_2}) .
$$
\end{theorem}

\begin{rem}
The idea of the proof of this theorem is borrowed from \cite{Mi}.
\end{rem}

\begin{proof}
Let $\O_\al$ be the domains constructed in Theorem \ref{t22}.
Put
$$
f_\al = \left.(({\cal L} - i I) X)\right|_{\O_\al\cap B_{2\rho}} , \quad 
f_\al \in L_2 (\O_\al\cap B_{2\rho}, \C^8; \er, \mu) .
$$
Let ${\cal L}_\al$ be the self-adjoint operator, 
defined by formula \eqref{08} on the domain 
${\cal D} (\O_\al\cap B_{2\rho})$.
Introduce
$X_{1,\al} = ({\cal L}_\al - i I)^{-1} f_\al \in {\cal D} (\O_\al\cap B_{2\rho})$.
We have
$$
\|X_{1,\al}\|_{{\cal D} (\O_\al\cap B_{2\rho})} = 
\|f_\al\|_{L_2(\O_\al\cap B_{2\rho})} 
\le \|X\|_{{\cal D} (\O)}  \quad \forall \ \al\in(0,1) .
$$
Let us consider the field 
$$
X_{2,\al} : = \left.X\right|_{\O_\al \cap B_{2\rho}} - X_{1,\al} , \quad 
\|X_{2,\al}\|_{{\cal D} (\O_\al\cap B_{2\rho})} \le 2 \|X\|_{{\cal D} (\O)} .
$$
Without loss of generality one can assume that 
$X_{2,\al_k} \mathop{\to}\limits_{\al_k \to 0} X_3$ 
(see footnote to Lemma \ref{l51})
weakly in the space
\begin{equation}
\label{515}
F(\O_\be\cap B_{2\rho}, \er) \oplus W_2^1 (\O_\be\cap B_{2\rho}, \C) 
\oplus F (\O_\be\cap B_{2\rho}, \mu) \oplus W_2^1 (\O_\be\cap B_{2\rho}, \C) 
\end{equation}
for any fixed $\be > 0$, and 
$$
X_3 \in F(\O\cap B_{2\rho}, \er) \oplus W_2^1 (\O\cap B_{2\rho}, \C) 
\oplus F (\O\cap B_{2\rho}, \mu) \oplus W_2^1 (\O\cap B_{2\rho}, \C).
$$
Let us show that the components of the field
$$
X_3 = \left( \begin{array}{cc} E_3 \\ \ph_3 \\ H_3 \\ \eta_3\end{array} \right)
$$
satisfy appropriate boundary conditions, 
i.e. $X_3 \in {\cal D} (\O\cap B_{2\rho})$.
Indeed, let $h \in L_2(\O\cap B_{2\rho}, \C^3)$, 
$\rot h \in L_2 (\O\cap B_{2\rho}, \C^3)$.
By virtue of Lemma \ref{l51} we have
\begin{eqnarray*}
( E_3,\rot h)_{L_2(\O\cap B_{2\rho})}-(\rot E_3,h)_{L_2(\O\cap B_{2\rho})} \\ 
= \mathop{\lim}\limits_{\al_k\rightarrow 0}
\left((E_{2,\al_k},\rot h)_{L_2(\O_{\al_k}\cap B_{2\rho})} - 
(\rot E_{2,\al_k},h)_{L_2(\O_{\al_k}\cap B_{2\rho})}\right) \\
= \mathop{\lim}\limits_{\al_k\rightarrow 0}
\left(( E,\rot h)_{L_2(\O_{\al_k}\cap B_{2\rho})} - 
(\rot E,h)_{L_2(\O_{\al_k}\cap B_{2\rho})}\right) \\
= (E,\rot h)_{L_2(\O)}-(\rot E,h)_{L_2(\O)} = 0, 
\end{eqnarray*}
since fields $X_{1,\al}$ and $X$ satisfy the conditions
$\left.\left(E_{1,\al}\right)_\tau\right|_{\dd(\O_\al\cap B_{2\rho})} = 0$
and
$\left.E_\tau\right|_{\dd\O} = 0$ respectively.
Thus, $\left.(E_3)_\tau\right|_{\dd(\O\cap B_{2\rho})} = 0$.
One can establish similarly the equalities
$\left.\left(\mu H_3\right)_\nu\right|_{\dd(\O\cap B_{2\rho})} = 0$ and
$\left.\eta_3\right|_{\dd(\O\cap B_{2\rho})} = 0$.
Finally,
$\left.\left(({\cal L}_{\al_k} - i) X_{1,\al_k}\right)\right|_{\O_\be\cap B_{2\rho}} 
= f_\be$
for $\al_k \le \be$. 
Therefore,
$$
\left.\left(({\cal L} - i) X_{2,\al_k}\right)\right|_{\O_\be\cap B_{2\rho}} =0 \ 
\text{ for } \ \al_k\le\be,
$$
where ${\cal L}$ is understood as differential expression \eqref{08}.
Hence, $({\cal L} - i) X_3 = 0$ in $\O$.
Since $X_3 \in {\cal D} (\O)$, we conclude that $X_3 = 0$.
Therefore, $X_{1,\al}$ converges to $X$ weakly 
in the space \eqref{515} for all $\be > 0$.

Further, introduce the vector 
$X_\al=\chi_{{\rho}}X_{1,\al}\in {\cal D} (\O_\al\cap B_{2\rho})$. 
Here $\chi_{{\rho}}\in C_0^\infty(\R^3), $
\[\chi_\rho(x)=\left\lbrace
\begin{array}{cc}
1,& x\in B_{\rho/2}\\
0,&x\notin B_\rho\\
\end{array}. \right.
\]
On the boundary of the set $\O_\al\cap B_{2\rho}$ 
the function $X_\al$ can be different from zero only on $\dd\O_\al$.
By Theorem \ref{localsmooth} $X_{\al} \in {\cal A} (\O_\al)$.
According to Theorems \ref{t35} and \ref{t38} 
$$
\|X_{\al}\|_{W_2^1 (\O_\al)} \le  
C(\rho,\er,\mu) \|X_{\al}\|_{{\cal D} (\O_\al)}  
\le C(\rho,\er,\mu) \|X\|_{{\cal D} (\O)} .
$$
Thus, there exists a sequence $\al_k \to 0$ 
and the field $X_0 \in W_2^1 (\O, \C^8)$ such that
$X_{\al_k} \to X_0$ weakly in $W_2^1 (\O_\be, \C^8)$ 
(see footnote to Lemma \ref{l51}) for all $\be > 0$;
$\|X_0\|_{W_2^1(\O)} \le C(\rho,\er,\mu) \|X\|_{{\cal D} (\O)}$.
So, $X = X_0 \in {\cal A} (\O)$.
\end{proof}

\subsection{($W^2_3\cap W^1_\infty$)-diffeomorphisms}
\begin{lemma}
\label{l52}
Let $\Omega$, $\tilde\Omega$ be bounded domains in $\R^3$,
$\psi: \Omega \to \tilde\Omega$ be a bijection such that
$$
\psi \in W^2_3 \cap W^1_\infty (\Omega), \quad
\psi^{-1} \in W^2_3 \cap W^1_\infty (\tilde\Omega).
$$
Let $u$ and $v$ be functions, connected via the relation
\begin{equation}
\label{52}
u_j (x) = \dd_j \psi_k (x) v_k (\psi(x)).
\end{equation} 
Let $s$ be a matrix-function.  
Denote $J_{jk}(x)= \dd_j \psi_k (x), y=\psi(x).$
Then for a.e. $x\in\O$
$$
(\rot_x u)(x)=(|\det J| (J^{-1})^t \rot_y v)(y),
$$
$$
(\div_x(su))(x)=
\left(|\det J|\div_y\left(\dfrac{J^tsJ v}{|\det J|}\right)\right)(y).
$$
\end{lemma}

\begin{proof}
We have
$$
(\rot_x u)_i=[\n_x,u]_i=
\<e_i,[J\n_y,Jv]\>=|\det J|\<J^{-1}e_i,[\n_y,v]\>
=|\det J| ((J^{-1})^t \rot_y v)_i.
$$
Here in the third equality the derivatives of $J$ cancel out, 
since  
$$
\dd_{x_j} J_{km} v_m - \dd_{x_k} J_{jm} v_m = 
(\dd_j \dd_k \psi_m - \dd_j \dd_k \psi_m) v_m = 0.
$$
Using the definition of the divergence in terms of 
distributions, we have
\begin{eqnarray*}
\intop\limits_\O\div_x(su)\eta dx=-\intop\limits_\O \< su,\n_x\eta\> dx
=-\intop\limits_{\tilde{\O}} \<sJv,J\n_y\eta\>|\det J|^{-1}dy\\
=\intop\limits_{\tilde{\O}}\div_y\left(\dfrac{J^tsJv}{|\det J|}
\right)\eta dy
=\intop\limits_\O|\det J|\div_y\left(\dfrac{J^tsJv}{|\det J|}\right)\eta dx, 
\quad \forall \eta \in C_0^\infty(\O, \R). \quad\qedhere
\end{eqnarray*}
\end{proof}

\begin{theorem} 
\label{t54}
Let $\Omega$, $\tilde\Omega$ be bounded domains in $\R^3$,
$\psi: \Omega \to \tilde\Omega$ be a bijection such that
$
\psi \in W^2_3 \cap W^1_\infty (\Omega)$, 
$\psi^{-1} \in W^2_3 \cap W^1_\infty (\tilde\Omega)$.
The transformation \eqref{52} maps the space $F(\Omega, \er, \tau)$
to the space $F(\tilde\Omega, \tilde\er, \tau)$,
$F(\Omega, \mu, \nu)$ to $F(\tilde\Omega, \tilde\mu, \nu)$,
$W_2^1(\Omega, \tau)$ to $W_2^1(\tilde\Omega, \tau)$,
and $W_2^1(\Omega, \mu, \nu)$ to $W_2^1(\tilde\Omega, \tilde\mu, \nu)$.
Here the matrix-functions
$$
\tilde{\er}(y)=\dfrac{J^t \er(x) J}{|\det J|},\quad 
\tilde{\mu}(y)=\dfrac{J^t \mu(x) J}{|\det J|}
$$ satisfy \eqref{01} and \eqref{04} simultaneously with $\er$ and $\mu$,
$y=\psi(x)$, $J_{jk}(x)= \dd_j \psi_k (x)$.
Moreover, the norm estimates
$$
c_1 \|u\|_{F(\Omega, s)} \le
\|v\|_{F(\tilde\Omega,\tilde s)} \le
c_2 \|u\|_{F(\Omega, s)}, \quad s = \er\ \text{or}\ \mu,
$$
$$
c_1 \|u\|_{W^1_2(\Omega)} \le
\|v\|_{W^1_2(\tilde\Omega)} \le
c_2 \|u\|_{W^1_2(\Omega)}
$$
take place.
Here $c_1$ and $c_2$ depend only on the norms 
of $\psi$ and $\psi^{-1}$ in $W_3^2$ and $W^1_\infty$ .
\end{theorem}

\begin{proof}
Let $u$ and $v$ satisfy the relation \eqref{52}, $u\in F(\O,s)$, $s=\er$ or $\mu$. 
Lemma \ref{l52} implies that $v\in F(\tilde{\O},\tilde{s})$, 
and two-sided estimates take place.

Further, let 
$$
u\in F(\O,\er,\tau), \quad h\in L_2(\O,\C^3), \quad 
\rot h\in L_2(\O,\C^3), \quad \tilde{h}(y) = J(x)^{-1} h(x) .
$$ 
Then
\[
(\rot_y v,\tilde h)_{L_2(\tilde\O)} 
= (\rot_x u,h)_{L_2(\O)} = (u,\rot_x h)_{L_2(\O)}
= (v,\rot_y \tilde{h})_{L_2(\tilde\O)},
\]
Therefore, $\left.v_\tau\right|_{\dd \tilde{\O}}=0$. 
Similarly, if $u\in F(\O,\mu,\nu), \ph\in W^1_2(\O,\C)$, then
\[
(\div_y(\tilde{\mu}v),\ph)_{L_2(\tilde\O)}
=(\div_x(\mu u),\ph)_{L_2(\O)} 
=-(\mu u,\n_x\ph)_{L_2(\O)} 
=-(\tilde{\mu}v,\n_y\ph)_{L_2(\tilde\O)},
\]
which means that $\left.(\tilde{\mu}v)_\nu\right|_{\dd \tilde{\O}}=0$.
Finally, multiplication by $J\in W^1_3\cap L_\infty$ 
is the bounded operator in Sobolev space $W^1_2$.  
\end{proof}

\subsection{Localization}
It is easy to see that multiplication by smooth bounded function 
does not move elements of the space ${\cal D} (\O)$ from this space 
(see e.g. \cite{BS87smzh}).

\begin{lemma}
\label{l53}
Let $\O \subset \R^3$, $X \in {\cal D} (\O)$, $\ze \in C_0^\infty (\R^3)$.
Then
$$
\ze X \in {\cal D} (\O), \qquad 
\|\ze X\|_{{\cal D} (\O)} \le C 
\left(\|\n\ze\|_{L_\infty} + \|\ze\|_{L_\infty}\right) \|X\|_{{\cal D} (\O)} .
$$
\end{lemma}

{\it Proof of Theorem \ref{t11}.}
Let $\O$ be a bounded domain, $\O \in \mathcal{C}(W_3^2\cap W^1_\infty)$. 
By Definition \ref{o03} the boundary $\dd\O$ can be covered 
by finite number of domains $D_j$, such that for each of them 
there exists a corresponding diffeomorphism $\psi_j$. 
This collection of domains may be completed 
up to covering of the whole domain $\O$. 
Denote the new covering as $ \mathcal{W}_j$ 
and denote the subordinate partition of unity as $\ze_j$:
$$
\label{razb}
\ze_j \in C_0^\infty (\R^n), \ \ 0 \le \ze_j (x) \le 1, \ \ 
\supp \ze_j \subset \mathcal{W}_j, \ \ 
\sum_j \ze_j (x) = 1 \ \text{for} \ x \in \overline\O.
$$
Let 
$$
X = \left( \begin{array}{cc} E \\ \ph \\ 
H \\ \eta\end{array} \right)\in {\cal D} (\O)
\quad \text{and} \quad 
X_j = \ze_j X = \left( \begin{array}{cc} 
E_j \\ \ph_j \\ H_j \\ \eta_j \end{array} \right) .
$$
By Lemma \ref{l53} 
$$
X_j \in {\cal D} (\O), \quad 
 \|X_j\|_{{\cal D} (\O)} \le C \|X\|_{{\cal D} (\O)} \quad \text{and} 
\quad \supp X_j \subset \mathcal{W}_j.
$$
Put
$$
Y_j (y) = \left( \begin{array}{cc}
J_j(x)^{-1} E_j(x)\\ \ph_j(x)\\
J_j(x)^{-1} H_j(x)\\ \eta_j(x) \end{array} \right), 
$$
where the variables $x$ and $y$ are connected by the formula $y = \psi_j (x)$,
$J_j$ is the Jacobi matrix of mapping $\psi_j$.

By virtue of Theorem \ref{t54} 
$Y_j\in \mathcal{D}(\psi_j(\mathcal{W}_j)\cap V_j)$, 
where $V_j$ is the corresponding special Lipschitz domain. 
Since $\supp Y_j\subset \psi_j(\mathcal{W}_j)$, 
the vector $Y_j$ belongs to the space $\mathcal{D}( V_j)$. 
Furthermore, due to Theorem \ref{t52} $Y_j \in {\cal A} (V_j)$.
Applying Theorem \ref{t54} once again, we get 
$X_j \in {\cal A} (\mathcal{W}_j\cap \O)$ and
$$
\|X_j\|_{W_2^1 (\O)} \le C_1 \|Y_j \|_{W_2^1 (V_j)}
\le C_2 \|Y_j \|_{{\cal D} (V_j)} \le C_3 \|X_j\|_{{\cal D} (\O)}.
$$
Hence,
$$
X = \sum_j X_j \in {\cal A} (\O) \quad \text{and} \quad 
\|X\|_{W_2^1 (\O)} \le \sum_j \|X_j\|_{W_2^1 (\O)} 
\le C_4 \sum_j \|X_j\|_{{\cal D} (\O)}  \le C_5 \|X\|_{{\cal D} (\O)} .\quad
\qed
$$

%

\end{document}